\documentclass[superscriptaddress,amsmath,amssymb,twocolumn]{revtex4-1}
\pdfoutput=1

\usepackage{amsthm}
\usepackage{graphicx}
\usepackage{hyperref}

\usepackage[usenames,dvipsnames]{color}

\newcommand{\mc}[1]{\mathcal #1}
\newcommand{\new}{} 

\newcommand{\bra}[1]{\langle #1 |}
\newcommand{\ket}[1]{| #1 \rangle}

\newcommand{\ketbra}[2]{\ket{#1}\bra{#2}}
\newcommand{\proj}[1]{\ket{#1}\bra{#1}}
\newcommand{\tr}{{\rm Tr}\,}
\newcommand{\one}{{\bf 1}}

\newcommand{\alg}[1]{\mathcal #1}
\newcommand{\hil}[1]{\mathcal #1}
\newcommand{\chan}[1]{\mathcal #1}

\newcommand{\comp}[1]{\widehat{#1}}
\newcommand{\comph}[1]{\widehat{#1}^\dagger}
\newcommand{\lcomp}[2]{\widehat{#2}^{#1}}
\newcommand{\lcomph}[2]{(\widehat{#2}^{#1})^\dagger}

\newcommand{\cchan}[1]{{\comp{\chan #1}}}
\newcommand{\lcchan}[2]{{\lcomp{#1}{\chan #2}}}
\newcommand{\lcchanh}[2]{{\lcomph{#1}{\chan #2}}}

\newcommand{\chanh}[1]{\chan{#1}^\dagger}

\newcommand{\id}{{\rm id}}

\newcommand{\ops}{{\mathcal B}}

\newcommand{\V}{W}

\newcommand{\be}{\begin{equation}}
\newcommand{\ee}{\end{equation}}
\newcommand{\bs}{\begin{split}}
\newcommand{\es}{\end{split}}

\newcommand{\mytheorem}{Theorem} 
\newcommand{\mylemma}{Lemma}
\newcommand{\mycorollary}{Corollary}
\newcommand{\myexample}{Example}

\theoremstyle{plain}
\newtheorem{proposition}{Proposition}
\newtheorem{theorem}[proposition]{\mytheorem}
\newtheorem{lemma}[proposition]{\mylemma}
\newtheorem{corollary}[proposition]{\mycorollary}
\newtheorem{example}[proposition]{\myexample}

\theoremstyle{definition}
\newtheorem{definition}{Definition}

\begin{document}

\title{Approximate recovery with locality and symmetry constraints}
\author{C\'edric B\'eny}
\address{Department of Applied Mathematics, Hanyang University (ERICA), 55 Hanyangdaehak-ro, Ansan, Gyeonggi-do, 426-791, Korea.}
\author{Zolt\'an Zimbor\'as}
\address{Department of Theoretical Physics, Wigner Research Centre for Physics, Hungarian Academy of Sciences, P.O. Box 49, H-1525 Budapest, Hungary}
\author{Fernando Pastawski}
\address{Dahlem Center for Complex Quantum Systems, Freie Universit\"at Berlin, 14195 Berlin, Germany  and\\
 Gravity, Quantum Fields and Information, Max  Planck  Institute  for  Gravitational  Physics, Am M\"uhlenberg 1, D-14476 Potsdam-Golm, Germany}

\begin{abstract}
Numerous quantum many-body systems are characterized by either fundamental or emergent  constraints---such as gauge symmetries or parity superselection for fermions---which effectively limit the accessible observables and realizable operations. 
Moreover, these constraints combine non-trivially with the potential requirement that operations be performed locally. 
The combination of symmetry and locality constraints influence our ability to perform quantum error correction in  two counterposing ways.
On the one hand, they constrain the effect of noise, limiting its possible action over the quantum system.
On the other hand, these constraints also limit our ability to perform quantum error correction, or generally to reverse the effect of a noisy quantum channel. 
We analyze the conditions that local channels should satisfy in the constrained setting, and characterize the resulting optimal decoding fidelity. In order to achieve this result, we introduce a concept of local complementary channel, and prove a new local information-disturbance tradeoff.
\end{abstract}

\maketitle

\section{Introduction}

A usual assumptions in quantum information theory literature, is that all self-adjoint operators on a Hilbert space can be in principle measured. 
However, in many systems of interest, such as fermions or gauge fields, accessible observables are limited to certain subalgebras. 
Fermions provide the simplest example for this; only observables commuting with fermion parity are considered physical. 
More generally, gauge theories, which describe the dynamics of elementary particles with remarkable accuracy and elegance, postulate that physical observables must commute with the gauge constraints.
Finally, there are many-body physical systems for which these symmetries are not fundamental but emergent \cite{Wen2004, Zeng2015}, the resulting effect is equivalent as long as the energy of the environments and/or observers is reasonably limited.

Despite the ubiquitous nature of gauge symmetries, superselection rules and locality constraints, a general quantum information framework for studying the interplay of these constraints is in its infancy.
The pivotal role that quantum information and particularly quantum error correction (QEC) is playing in recent developments of both condensed matter and high energy physics urgently demands the development of a solid framework.
Significant progress has been made reinterpreting  entanglement in the presence of superselection rules and gauge constraints~\cite{schuch2004, Banuls2007,donnelly2012,casini2014,vanacoleyen2016,Witten2018}. 
Here, we consider the information-disturbance tradeoff and its application to  QEC, where it can be used to characterize which communication channels (representing noise) can be reversed on a given code.
In particular, we will mention two areas where this will immediately prove useful.

On the  condensed-matter side, a plethora of work has recently been dedicated to the classification of symmetry protected topological phases (SPTs) \cite{Chen2011, Chen2013, Senthil2015} and symmetry enriched topological phase (SETs).
These generalize the notion of topological order to a settings where a symmetry or gauge constraints are imposed.
A SPT phase is a phase which would be trivial if the imposed symmetry were allowed to be broken yet become ``disconnected'' from the trivial (product) phase if the symmetry is imposed.
In contrast, SET are phases disconnected from the trivial phase even if the symmetry constraint is lifted.
As the stability for topological order essentially requires QEC conditions  \cite{Bravyi2010, Michalakis2013} , it is only natural that this connection extend to the symmetry protected setting.
\new{Indeed, the seminal example for an SPT phase, to which we apply our results, is Kitaev's Majorana chain \cite{Kitaev2001}.
This is a gapped phase with topological degeneracy protected by fermionic parity conservation and geometric locality. 
In parallel to the theoretic exploration of new candidate symmetry protected codes for quantum information storage \cite{Bravyi:2010de, Roberts2018}, Majorana chains are being pursued experimentally as a candidate qubit \cite{Mourik2012}.}

Within high energy physics, recent progress in holography  provides the second natural arena for our results.
In particular, the realization that the bulk/boundary mapping in holography presents properties of a QEC \cite{Almheiri2014} has lead to vigorous debate with respect to the role of symmetries and gauge constraints.
Whereas qualitative features have been reproduced in the context of traditional QEC theory ~\cite{pastawski2015,pastawski2016, Hayden2016, Cotler2017}, it has been argued that gauge constraints play an essential role \cite{Mintun2015,Pastawski2017b}.
In fact, the boundary theory in holography is a gauge quantum field theory and thus the validity QEC assumptions must be re-examined in the presence of corresponding constraints.
 
In a general approach to constrained systems, the allowed observables form a $*$-algebra $\alg A$, namely the set is closed under multiplication and the operation of taking the adjoint. Moreover, the observables local to a certain region of space form a $*$-subalgebra of $\alg A$ which is not necessarily associated with a tensor factor of the full Hilbert-space. 

Technically, if they are infinite-dimensional, these algebras require some additional mathematical structure, such as that of a C$^*$-algebra or von Neumann algebra, but here, for simplicity, we consider only algebras of finite-dimensional matrices, closed under the conjugate-transpose ($\dagger$-algebras) where those concepts are all equivalent. This is appropriate for systems of fermions on finite lattices, but will require some generalization to be applicable to lattice gauge theory to account for the fact that the Hilbert space of the gauge field on each edge is not finite-dimensional for Lie groups.  

A natural starting point would be the {\it operator-algebra quantum error correction} (OAQEC)~\cite{Beny2007} (a synthesis of the theory of noiseless subsystems and subsystem codes~\cite{knill2000,Kribs2005}) because it provides sufficient conditions for a quantum channel to be reversible on a given code when one only cares about a given $\dagger$-algebra of observable.
We instead consider a broad generalization of this approach to the approximate setting introduced in Ref.~\cite{beny2010} (\mytheorem~\ref{thmduality}), based on techniques borrowed from Ref.~\cite{kretschmann2008} (information-disturbance tradeoff). 
This approach is a particular formalization of the general fact that quantum information can be recovered after the action of a channel if and only if it is not available in the environment (as characterized by the complementary channel).

In section \ref{rev_const}, we show that these results can be adapted to the case where the recovery map is required to be ``physical'', in that it does not reveal information outside of the allowed observables (\mycorollary~\ref{main1}). 
Indeed, constraints on physical observables also affect channels as these should not enable the indirect measurement of unphysical observables. 
In addition, if we require that our channels be acting locally to some region of space, then they must leave unchanged all the observables acting outside that region.

In section \ref{Sec:RevLocal}, we further extend them to a situation where the recovery map is required to be local, in that it fixes observables associated with a complementary region of space (\mytheorem~\ref{main2a} and \mycorollary~\ref{main2b}). 
This requires a concept of {\em local complementary map}, defined in Section~\ref{lcc}.

\section{Preliminaries}
In this section we review material required to present and illustrate our results.
In particular section \ref{sec:PreFermion} reviews the algebra for fermions which allows providing the simplest examples beyond tensor product Hilbert spaces with a genuine physical motivation.
In section \ref{sec:PrePhysicalChannel}, we propose a notion of physicality of a quantum channel as derived form the physicality from the algebras of observables.
Finally, we review the main result of Ref. \cite{beny2010,beny2011} and exemplify how they comprise traditional QEC conditions.
The current work can be seen as a natural generalization of these result to a setting where symmetry constraints and locality are imposed on the channels involved.

\subsection{Fermions\label{sec:PreFermion}}

As an example system with a constraint, let us consider a system of spinless fermions on a lattice with sites indexed by $\Omega = \{1,\dots,N\}$. 
One associates to each site an annihilation operator $a_i$. These operators generate a minimal $\dagger$-algebra such that $a_i^\dagger a_j + a_j a_i^\dagger = \delta_{ij} \one$ and $a_i a_j + a_j a_i = 0$ hold for all $i$, $j$. 
Before considering locality, this algebra is isomorphic to that of all operators acting on a Hilbert space $\hil H$ of dimension $2^N$ (which can be regarded as the Fock space corresponding to $N$ modes with a basis $(a_N^{n_N})^\dagger \cdots (a_1^{n_1})^\dagger \ket 0$, where $n_1,\dots n_N \in \{0,1\}$ count the ``number of fermions'' at each sites, and $\ket 0$ is the Fock vacuum).

For any region of space  corresponding to the subset of vertices $\omega \subseteq \Omega$, we want to interpret the $\dagger$-subalgebra generated by the operators $a_i$ for $i \in \omega$ as characterizing the observables local to $\omega$. These algebras, however, do not commute for disjoint subsets. In order to make sure that observables in disjoint regions are jointly measurable, and hence commute, we declare that only observables which are even order polynomials in the annihilation operators are {\em physical}. 

This is equivalent to saying that the physical operators are those that commute with the parity observable $C$ which has eigenvalues $1$ and $-1$ respectively for states with an even and odd number of fermions, and hence is referred to as the \emph{parity superselection rule} \cite{Banuls2007, zimboras2014}.

Hence, for every region of space $\omega \subseteq \Omega$, we assign a physical subalgebra $\alg A_\omega$ of operators which are functions of the operators $a_i$, $i \in \omega$, and commute with $C$. Specifically, these algebras have the form $\alg A_\omega \simeq \alg M_{1} \oplus \alg M_{-1}$, where $\alg M_{\pm 1}$ are full matrix algebras of dimensions $k$ by $k$ with $k = 2^{|\omega|-1}$. They correspond to all operators acting on the Hilbert space sectors with an even, respectively odd, number of fermions on region $\omega$. 
This is an instance of a {\em local quantum theory} as defined in algebraic quantum field theory (but simpler since space is discrete).

\new{
The fermionic operators $a_j$ may also be formally expressed in terms of their Majorana counterparts $w_{2j} = {a_j + a^\dagger_j}$ and $w_{2j-1} = {ia_j - i a^\dagger_j}$.
These Majorana operators are manifestly Hermitian ($w_k = w^\dagger_j$) satisfy the fermionic anti-commutation relations $\{w_k, w_l\}_+ =2 \delta_{lm}$.
When convenient, we will chose to index the system in terms of twice as many Majorana indices. 
In terms of the Majorana indices, we may define the parity observable 
\begin{align}\label{eq:parity}
	C_\omega = i^{|\omega|-1}  \prod_{j \in  \omega}  w_j
\end{align}
which is only defined up to a global sign, given that exchanging the order of two Majorana operators introduces an additional minus sign.
Note that $C_\omega$ will itself be a physical observable effective at defining superselection sectors precisely when $|\omega|$ is finite and even but not in general.
Finally, we may define the projectors onto the parity superselection sectors as $P_\pm = \frac{\one \pm C}{2}$.}

\subsection{Physical channels \label{sec:PrePhysicalChannel}}

Let us consider a channel, also known as completely-positive and trace-preserving (CPTP) maps, $\chan N: \ops(\hil H) \rightarrow \ops(\hil K)$, where $\hil H$ and $\hil K$ are finite-dimensional Hilbert spaces, and $\ops(\hil H)$ denotes the set of all operators on $\hil H$. 
Also let $\alg A$ and $\alg B$ be $\dagger$-algebras of operators acting on $\hil H$ and $\hil K$ respectively, which represent the physical observables. Below, we always assume that these algebras contain the identity on their respective Hilbert spaces (in general, an algebra's identity element could be a projector on the Hilbert space). 

For $\chan N$ to be physical, it should be such that the recipient cannot gain information about unphysical observables. This is most easily expressed in the Heisenberg picture as $\chanh N(\alg B) \subseteq \alg A$, where $\chanh N$ is defined by $\tr(\chan N(\rho)B) = \tr(\rho \chanh N(B))$ for all operators $B$ on $\hil K$ and all density matrices $\rho$ on $\hil H$. 

Let us introduce the Hilbert-Schmidt orthogonal projector $\chan P$ on $\alg A$ (which is a {\em conditional expectation} from $\ops(\hil H)$ to $\alg A$). This is a channel satisfying $\chanh P = \chan P = \chan P^2$ whose range is $\alg A$. Similarly, let $\chan Q$ be the projector on $\alg B$.
\new{ For instance, the conserve parity $C$ of fermions, this projector is given by $\chan P(\rho) = \frac{1}{2}[\rho + C \rho C]$ corresponding to its interpretation as dephasing w.r.t. $C$.
It may equivalently be written as $\chan P(\rho) = P_+ \rho P_+ + P_- \rho P_-$ corresponding to an interpretation as "blind measurement" where $P_+$ and $P_-$ are projectors onto the even and odd parity sectors respectively.}

\begin{definition}\label{def:PhysicalChannel}
Let $\chan P$ and $\chan Q$ be the projector channels onto the physical algebra of observables on the source and target Hilbert space of $\chan N$ respectively. 
We say that the channel $\chan N$ is physical with respect to this restriction if  
\begin{equation}\label{eq:PhysicalChannel}
\chan Q \chan N \chan P = \chan Q \chan N.
\end{equation}
\end{definition}
This definition is central to our result as it allows naturally incorporating locality and symmetry conditions into the setting of channels.

\begin{proposition}\label{Fermionic_physical_channels} 
\new{Let $C$ be the charge observables of a fermionic system.
Then $\chan N$ is a physical channel on this space if it admits a Kraus representation 
\begin{equation}
	\chan N(\rho) = \sum_j E_j \rho E^\dagger_j,
\end{equation}
where all $E_j$ have definite parity (i.e. $E_j C = \pm C E_j$).
Conversely, if $\chan N$ is physical, then the channel $\chan Q \chan N = \chan Q \chan N \chan P$, which has equivalent action on the physical observables, can be written in terms of definite parity Kraus operators.}
\end{proposition}
\begin{proof}
	\new{
To see that all channels of this form are indeed physical, it is sufficient to expand the Kraus operators in equation \eqref{eq:PhysicalChannel}.
Indeed, expanding $\chan Q \chan N \chan P$, we obtain Kraus operators $\{E_j, C E_j, E_j C, C E_j C\}$ with normalization $1/2$ whereas expanding $\chan Q \chan N$ we obtain Kraus operators $\{E_j, C E_j \}$ with normalization $1/\sqrt{2}$.
We may then use $C^2=1$ and the (anti-)commutation relation to obtain $C E_j C \rho C E^\dagger_j C = E_j \rho E^\dagger_j$ and $E_j C \rho C E^\dagger_j = C E_j \rho E^\dagger_j C$.
Conversely, given the channel $\chan Q \chan N \chan P$ and its Kraus operators $\{E_j, C E_j, E_j C, C E_j C\}$, we may rewrite it in terms of Kraus operators commuting with $C$, $\{E_j + CE_jC, C E_j + E_j C\}$ and operators anti-commuting with $C$, $\{E_j - C E_j C, C E_j - E_j C\}$.}
\end{proof}
\new{
Note that some channels as $\chan N(\rho) = \frac{1}{4}[ (1+iw)\rho(1-iw) + (1-iw) \rho (1+iw)]$ are not manifestly physical, but admit an equivalent representation $\chan N(\rho) = \frac{1}{2}[ \rho + w \rho w]$ which is, and may be thought of as the dephasing channel w.r.t. a Majorana mode $w$.}

\subsection{Reversal and information-disturbance trade-off}

Given a noise channel $\chan N$, an important question is whether the effect of this channel can be reversed, i.e., whether there is a recovery channel $\chan R$ such that $\chan R \chan N(\rho) = \rho$ for a certain set of states $\rho$, typically all those supported on a certain subspace. Alternatively, one may ask whether the channel can be reversed on a single state $\rho$, but with certain restriction as to the locality of $\chan R$. In the literature, variations of these questions have been referred to as {\em channel sufficiency}~\cite{petz1988}, {\em quantum error correction}~\cite{Knill1997}, {\em channel recoverability}~\cite{Junge2016} or {\em channel reversal}~\cite{buscemi2016}. 

In the presence of constraints, what we require is the weaker condition $\chan R \chan N(\rho) = \chan P(\rho)$, since we do not worry about the expectation value of unphysical observables.

In addition, if equality is not exactly achieved, we may want to quantify the error using some measure of similarity between channels.
Here we focus on the following result from Refs.~\cite{beny2010,beny2011}, where channels are compared using a fidelity $F$:
\begin{theorem}
\label{thmduality}
For any two channels 
$\chan N$ and $\chan M$,
\begin{equation}
\label{duality}
\max_{\chan R} F(\chan R \chan N, \chan M) 
= \max_{\chan S} F(\cchan N, \chan S \cchan M),
\end{equation}
where the maxima are taken over all CPTP maps, and $\cchan N$ and $\cchan M$ are any channels {\em complementary} to $\chan N$ and $\chan M$ respectively. 
This holds taking $F$ to be either (a) the {\it entanglement fidelity} $F_\rho$ or (b) the {\it worst-case entanglement fidelity}  $F_\V$.
\end{theorem}

Specifically, the two fidelity measures considered are defined as follows.
\begin{enumerate}
\item[(a)] The
{\it entanglement fidelity} compares the effect of two channels on a single state, while accounting for the possible loss of entanglement with a reference system:
\begin{equation}
F_\rho(\chan N, \chan M) := f((\chan N  \otimes \id)(\psi),(\chan M \otimes \id)(\psi)),
\end{equation}
where $\psi \equiv \proj \psi$ denotes any purification of $\rho$, and $f$ is the fidelity $f(\rho,\sigma) = \tr(\sqrt{\sqrt \rho \sigma \sqrt \rho})$. This quantity can also be used to bound the average fidelity with respect to an ensemble averaging to $\rho$~\cite{schumacher1996}.
\item[(b)]
{\it Worst-case entanglement fidelity}:
Alternatively, channels can be compared on a {\em code}, that is, a subspace $\hil H_0$ of $\hil H$ defined by a canonical isometry $\V: \hil H_0 \rightarrow \hil H$, and can be characterized using the worst-case entanglement fidelity
\begin{equation}\label{eq:worse_fidelity}
F_\V(\chan N, \chan M) := \min_\rho f((\chan N \chan \V \otimes \id)(\psi_\rho),(\chan M \chan \V \otimes \id)(\psi_\rho)),
\end{equation}
where $\chan \V(\rho) = \V \rho \V^\dagger$, and $\psi_\rho$ is any purification of $\rho$. 
\end{enumerate}

Both channel fidelities can be used to construct distances satisfying the triangle inequality, such as the Bures distance.
Note that if $\rho \in \chan B(\hil H_0)$, then by definition $F_W \leq F_\rho$ as suggested by the names.
Below, all we need is the fact that both fidelities are monotonic under the left action of any channel, i.e., 
\begin{equation}
F(\chan R \chan N, \chan R \chan M) \ge F(\chan N,\chan M)
\end{equation} 
for any channels $\chan R$, $\chan N$, $\chan M$.

A {\em complementary channel} $\cchan N$ of $\chan N$ can be built as follows. The Stinespring dilation theorem states that there is an isometry $V: \hil H \rightarrow \hil K \otimes \hil L$ such that $\chan N(\rho) = \tr_{\hil L} V \rho V^\dagger$, where $\tr_{\hil L}$ is the partial trace over $\hil L$. Let $\ket i$ denote elements of a basis of $\hil L$, then we obtain the Kraus operators $E_i = (\one \otimes \bra i) V$. 
Reciprocally, $V = \sum_i E_i \otimes \ket i$. 
Any such dilation gives us a complementary channel $\cchan N(\rho) = \tr_{\hil K} V \rho V^\dagger$. We also call a channel $\cchan N'$ complementary to $\chan N$ if there exists channels $\chan R$ and $\chan S$ such that $\cchan N' = \chan R \cchan N$ and $\cchan N = \chan S \cchan N'$, where $\cchan N$ has the above form. (This is the equivalence relation defined in Ref.~\cite{beny2011}).

In order to illustrate the use of theorem \ref{thmduality}  for channel reversal, we first present the setting of perfect recovery in traditional subspace QEC. 
\begin{example}[Subspace QEC]
Consider the case when both fidelities in Eq.~\eqref{duality}  are maximal using $F=F_\V$, and for $\chan M = \id$ (i.e. we wish $\chan R$ to recover {\em all} information initially available in the code defined by $\V$). 
In this case, we can use $\cchan M = \tr$. Hence the channel $\chan S$ to be optimized on the right hand side of Eq.~\eqref{duality} is just a state $\sigma$, since it is applied to the one-dimensional density matrix $1$: $\chan S(1) = \sigma$.
Eq.~\eqref{duality} means that $\chan N$ is exactly correctable on the code defined by $\V$ if and only if there exists a state $\chan \sigma$ such that $\cchan N(\V \rho \V^\dagger) = \sigma \, \tr(\rho)$.
\end{example}
In terms of an explicit expression for $\chan N$,
\begin{align}\label{eq:NoiseChannel}
\chan N(\rho) &= \sum_i E_i \rho E_i^\dagger  \text{ and }\\
\cchan N(\rho) &= \sum_{ij} \tr(\rho E_j^\dagger E_i) \ketbra i j,\label{eq:ComplementaryNoise}
\end{align}
this means that for all $i$, $j$, $\V^\dagger E_j^\dagger E_i \V = \bra i \sigma \ket j \V \V^\dagger$, which are the Knill-Laflamme conditions for quantum error correction  \cite{Knill1997}.

\begin{example}[OAQEC]
When $\chan M = \chan P_{\alg A}$ is the projector on a $\dagger$-algebra $\alg A$: the condition from maximum fidelity yields that $\alg A$ is correctable on the code defined by the isometric encoding $\chan \V(\rho) = \V \rho \V^\dagger$ if and only if 
\begin{equation}
\cchan N \chan \V = \chan S \chan P_{\alg A'} \chan \V
\end{equation}
for some channel $\chan S$, where we used the fact that a channel complementary to $\chan P_{\alg A}$ is $\cchan P_{\alg A} = \chan P_{\alg A'}$: the projector on the {\em commutant} of $\alg A$. 
\end{example}

To recover the original formulation of {\em operator algebra QEC} (OAQEC)~\cite{Beny2007,Beny2007b}, we use $\V = \one$, but replace $\chan N$ by $\chan N \chan \V'$ where now $\chan \V'$ is the encoding isometry. We obtain that there is a channel $\chan R$ such that $\chan R \chan N \chan \V' = \chan P_{\alg A}$ if and only if there is a channel $\chan S$ such that $\comp{\chan N \chan \V'} = \cchan N \chan \V' = \chan S \chan P_{\alg A'}$. It is easy to see that we can then use $\chan S = \cchan N \chan \V'$. The resulting condition is that the range of $(\cchan N \chan \V')^\dagger$ be inside $\alg A'$. 
 Expressed in terms of Kraus operators, this is the result of Ref.~\cite{Beny2007}. This also characterizes subsystem codes~\cite{Kribs2005} when the algebra is a factor.

\section{reversal on constrained systems \label{rev_const}}

In the present section, we address the question of channel reversal for constrained systems, and provide some instructive examples.

\subsection{Reversal and constrained systems}

In the presence of constraints, the problem with using the duality relation given by Eq.~\eqref{duality}  is that the optimization on the left hand side is over channels $\chan R$ which may not be physical. 

Recall that we defined $\chan P$ and $\chan Q$ as the channels projecting respectively on the source and target's physical algebra.
By substituting $\chan Q \chan N$ for $\chan N$ in Eq.~\eqref{duality}, we obtain 
\begin{equation}
\label{basic}
\max_{\chan R} F(\chan R \chan Q \chan N, \chan M) = \max_{\chan S} F(\comp {\chan Q \chan N}, \chan S \comp {\chan M}).
\end{equation}
The recovery channel $\chan R' = \chan R \chan Q$ is properly physical since $\chan P \chan R' \chan Q = \chan P \chan R \chan Q = \chan P \chan R'$.

But does this correspond to the optimization over all physical recovery maps? 
Let us specialize this to the case where $\chan M = \chan P \chan M$. For instance, this is the case if $\chan M$ is the projector on any subalgebra of $\alg A$. 

Suppose $\chan R$ is any physical recovery channel. Then because of the contractivity of the Bures distance and the fact that $\chan P^2 = \chan P$,
\begin{equation}
F(\chan P \chan R \chan Q \chan N, \chan P \chan M) = F(\chan P \chan R \chan N, \chan P \chan M) \ge F(\chan R \chan N, \chan P \chan M).
\end{equation}
Therefore, if $\chan R$ is any physical optimal recovery channel then so is $\chan R' = \chan P \chan R \chan Q$.  
We conclude that:
\begin{corollary}
\label{main1}
For any physical channel $\chan N$ from a system with physical algebra projector $\chan P$ to one with projector $\chan Q$, and any channel $\chan M$.
\begin{equation}
\label{optimalrecovery}
\max_{\text{$\chan R$ physical}} F(\chan R \chan N, \chan P \chan M) = \max_{\chan S} F(\comp {\chan Q \chan N}, \chan S \comp {\chan P \chan M}),
\end{equation}
where the optimization on the left hand side is over channels $\chan R$ which are physical, i.e., such that $\chan P \chan R \chan Q = \chan P\chan R$ and the right hand side optimization over channels $\chan S$ is unconstrained.
\end{corollary}

If $\chan P$ denotes the projector on $\dagger$-algebra $\alg A$, then $\cchan P$ can be taken as the projector on the {\em commutant} $\alg A'$~\cite{Beny2009,beny2011}.  
Moreover, $\comp{\chan P \chan M}(\rho) = (\cchan P \otimes \id_{E})(V \rho V^\dagger)$ where $V$ is the isometry from the Stinespring dilation of $\chan M$, and $E$ the environment, or ancilla for this distillation. The same can be done to obtain $\comp{\chan Q \chan N}$.

\mycorollary~\ref{main1} holds whether we replace $F$ by $F_\V$ or $F_\rho$. 
For instance, with $F = F_\V$ and $\chan M = \id$,
the left-hand side of Eq.~\eqref{optimalrecovery} is 
the worst-case fidelity of recovery for states within the code-space defined by $\V$. 

In contrast, the entanglement fidelity $F_\rho$ provides a lower bound \cite{schumacher1996} on how well recovery fares on average with respect to an ensemble represented by $\rho$.
This bound was invoked in Ref. \cite{Pastawski2017b}, to evaluate the QEC properties of a thermal CFT ensemble as, in this setting, it would be much better behaved than $F_W$ going to the setting of an infinite dimensional Hilbert space.

\subsection{Example: exact reversal for commuting constraints}

\label{ex_rev_sup}

Let us consider a case where the physical algebra takes the form
$\alg A = \ops(\mathbb C^{n_1}) \oplus \dots \oplus \ops(\mathbb C^{n_d})$, with each superselection sector characterized by a projector $P_i$ of rank $n_i$. There is a corresponding ``charge'' observable $C = \sum_i c_i P_i$, $c_i \neq c_j$. For instance, for a system of fermions, $C$ would be the parity observable.
Alternatively, this algebra may arise from requiring that observables commute with self-adjoint operators $L_i$ which all commute with each other, such as in an Abelian gauge theory.

The projector $\chan P$ on $\alg A$ represents a ``blind measurement'' of $C$:
\begin{equation}\label{eq:BlindMeasure}
\chan P(\rho) = \sum_i P_i \rho P_i.
\end{equation}
A Stinespring dilation of $\chan P$ is given by the isometry $\sum_i P_i \otimes \ket i$ where the extra system records the measurement outcome. It follows that
\begin{equation}
\cchan P(\rho) = \sum_i \tr(P_i \rho)  \proj i.
\end{equation}
This is the quantum-to-classical channel characterizing the measurement of $C$. 
Hence, the map $\chan S$ in Eq.~\eqref{optimalrecovery} prepares a quantum state depending on the classical outcome $i$ of the global charge measurement.

Let us take $\chan Q = \chan P$, $\chan N$ as in \eqref{eq:NoiseChannel} and $\chan M = \chan P$ in Eq.~\eqref{optimalrecovery}, meaning that we wish to recover {\it all} the {\it physical} information.
Using the dilation isometry $V = \sum_i E_i \otimes \ket i$, we have
\begin{equation}
\begin{split}
\comp{\chan P \chan N}(\rho) &= (\cchan P \otimes \id)(V \rho V^\dagger)\\
&= \sum_{jnm} \tr(P_j E_m \rho E_n^\dagger P_j) \, {\ketbra m n} \otimes \proj j,
\end{split}
\end{equation}

In this example, the map $\chan S$ is of the form
\begin{equation}
\chan S(\proj i) = \sum_{j} \sigma_{j|i} \otimes \proj j
\end{equation}
where 
\begin{equation}
\label{sigmas_constraints}
\sigma_{j|i} \ge 0 \quad \text{and} \quad \sum_j \tr \sigma_{j|i} = 1.
\end{equation}
 
Let us consider the implication of Eq.~\eqref{optimalrecovery} for exact reversal of $\chan N$ on a subspace defined by the isometry $\V$. For exact reversal, $F_\V$ equals to $1$ exactly when $F_\rho$ equals $1$, provided $\rho$ has full rank on the code space defined by $\V$. 

Here, both $\chan R$ and $\chan N$ act on $\hil H$ and the dilation $\psi$ of $\rho$ is defined on an extended Hilbert space $\hil H \otimes \hil J$. The channel is exactly correctable in this case when
\begin{equation}
(\comp{\chan P \chan N} \otimes \id)(\psi) = (\chan S\cchan P \otimes \id)(\psi).
\end{equation}
Or in other words, for all $j,n,m$,
\begin{equation}
\label{gkl0}
\V^\dagger E_n^\dagger P_j E_m \V = \sum_i \bra n \sigma_{j|i} \ket m \V^\dagger P_i \V.
\end{equation}
The existence of the states $\sigma_{j|i}$ satisfying the above equation and the constraints \ref{sigmas_constraints} is necessary and sufficient for the correctability of the channel $\chan N$ on the code $\V$. 

This leads to:
\begin{corollary}\label{cor:correctable_with_superselection}
A necessary condition for the channel $\chan N$ to be correctable on the code with isometry $\V$, for a system with superselection charge $C = \sum_i c_i P_i$, is that there exist complex numbers $c_{ijnm}$ such that
\begin{equation}\label{gkl}
\V^\dagger E_n^\dagger P_j E_m \V = \sum_i c_{ijnm} \V^\dagger P_i \V.
\end{equation}
\end{corollary}
This necessary condition also becomes sufficient if $[C,\V \V^\dagger] = 0$, i.e., if the charge observable $C$ commutes with the code projector $\V\V^\dagger$. 
Indeed, this is equivalent to $P_i \V\V^\dagger P_j = 0$ for $i \neq j$. 
In other words, Eq. \eqref{gkl} is also a sufficient condition if all states in the code respect the superselection criterion, which is natural in this context. 

In the unconstrained case, corresponding to $C = \one$, it is sufficient to just ask for proportionality between the left and right hand side. 
A similar simplification can be achieved if the encoding map outputs states of a given charge, i.e., $P_i W = W$ for some specific $i$.
In this case, eq. \eqref{gkl} simplifies to 
\begin{equation}\label{simple_gkl}
\V^\dagger E_n^\dagger P_j E_m \V = c_{jnm} \V^\dagger \V,
\end{equation}
which is much more reminiscent of the original Knill-Laflamme condition of Ref.~\cite{Knill1997}.

\subsection{Application: Majorana ring (physical)}\label{sec:MajoranaRing}

Encoding and processing quantum information in Majorana particles \cite{Kitaev2001} is one of the most actively pursued approaches to protect quantum information from decoherence. This is mainly due to the possibility of engineering condensed matter Hamiltonians \cite{Alicea2010, Mourik2012} whose degenerate ground state space (code-space) contains two Majorana modes (fractionalized fermions).
An idealized version of the Majorana wire is characterized by a Hamiltonian of the form  $H_{\text{maj}} = \sum_{j=1}^{N-1}  (w_{2j+1} + i w_{2j})(w_{2j+1} - i w_{2j})$. 
In addition to being composed of geometrically local terms, this Hamiltonian is quadratic in the fermionic operators and thus commutes with the imposed parity symmetry. 
The ground state is such that the fermionic modes with annihilation operators $w_{2j} + i w_{2j+1}$ are unoccupied.
Note however, that the Majorana mode $w_1$ and $w_{2N}$ do not appear in this Hamiltonian, there is necessarily a degeneracy which repeats itself throughout the full spectrum and can be traced to the degree of freedom representing presence or absence of a delocalized fermionic mode $ b = w_1 + i w_{2N}$. 
  
The caveat with this simplified picture of a Majorana chain is that it supports a single logical fermion $b$ (i.e. a two-dimensional logical Hilbert space with coherent superposition  forbidden by parity). 
A more interesting setting includes multiple ($k$) logical fermions $b_1, b_2, \ldots, b_k$ which are encoded in way into ($2k$) unpaired Majorana modes of the physical system indexed by $\omega = \{\omega_1, \omega_2, \ldots, \omega_{2k}\}$, which have alternating parity. 
These are the modes which would not appear in the description of the underlying caricature Hamiltonian. 
For instance, this can be realized in a situation where a $1D$ ring is decomposed into $2k$ intervals $I_1, \ldots I_{2k}$ ($I_j = [\omega_j+1, \omega_{j+1}-1]$) with neighboring intervals representing domains having competing pairing nature as in Fig. \ref{fig:MajoranaRing} (i.e. regions $iw_{2j} w_{2j+1}$ or $iw_{2j} w_{2j-1}$). 
\begin{figure}[htb]
	\centering
	\includegraphics[width=0.4\textwidth]{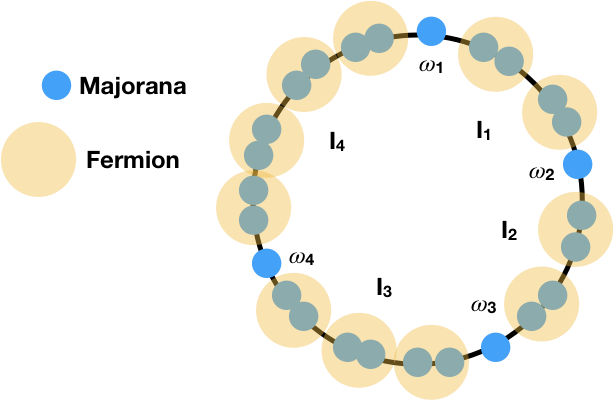}
	\caption{\new{The figure illustrates a how 13 fermions which have been fractionalized into 26 Majorana modes geometrically arranged on a circle.
	A Hamiltonian pairs  the Majorana modes into effective local fermions leaving four unpaired Majorana modes $\omega_j$.
These modes are expected to be well protected against naturally occurring noise forms with the protection initially increasing with the length of the intervals $I_j$.}}
	\label{fig:MajoranaRing}
\end{figure}

We can use the formalism developed so far to characterize all the correctable errors, when both the channel and noise are physical. 
Let us point out that such maps are allowed to change parity, by bringing fermions to or from the outside of our system, a phenomenon known as {\em quasiparticle poisoning} \cite{PhysRevB.85.174533, Temme2014}. 
In section~\ref{Sec:RevLocal}, we consider the situation where the channels are constrained to be local, and hence are not able to change the parity.

In the language of section \ref{rev_const}, we take $\chan P$  and $\chan Q$ as in Eq. \eqref{eq:BlindMeasure} to be the projector onto the physical $N$-fermion algebra. 
The encoding channel $\chan \V$ simply initializes the Majorana modes indexed by  $\bigcup_{j=1}^{2k} I_j \equiv \Omega \setminus \omega$ to the corresponding nearest neighbor pairing. 
This guarantees that $\chan \V = \chan \V^2$ and that orthogonal parity projectors on the full system are mapped by $\chan \V^\dagger$ onto orthogonal parity projectors on $\omega$ (i.e. $\V^\dagger C_\Omega \V = \pm C_\omega)$ up to a sign. 

\new{ To build intuition, let us consider some examples of correctable and non-correctable noise maps.
A simple unitary noise operator which does not satisfy eq. \ref{gkl} is $\chan N(\rho) = E_1 \rho E^\dagger_1$ with $E_1 = (1 + w_{\omega_1})/\sqrt{2}$.
However, this map is not a physical fermion map as required by proposition \ref{Fermionic_physical_channels}: more explicitly, ${\chan N}^\dagger(iw_{\omega_1}w_{\omega_2}) = -i w_{\omega_2}$ which is an odd monomial and thus, not a physical observable (we may also compare channels $\chan P^\dagger \chan N^\dagger \chan P^\dagger$ and $\chan N^\dagger \chan P^\dagger$ acting on this mode).
We may now multiply \eqref{gkl} by $P_-W$ on the left and by $W^\dagger P_+$ on the right.
The RHS will yield zero by assumption, whereas the LHS will contain a non zero contribution corresponding to a coherent superposition between two parity sectors.}

\new{In contrast, if the noise map consists of a single unitary Kraus operator with definite parity action, then such a noise map will be both physical and correctable by its inverse which is also physical.
Notably, this includes the case where the noise map performs a parity flip, changing the global physical parity.}

\new{A simple noise map $\chan N$ which is physical yet not correctable is given by the two Kraus operators $w_{\omega_1}$ and $w_{\omega_2}$, with normalization $1/\sqrt{2}$.
If there are more than two encoded Majorana modes, this map annihilates the non-trivial logical observable $M = iw_{\omega_2}w_{\omega_3}$ (i.e. this operator commutes with $\V\V^\dagger$).
The observable $M$ is physical yet irrecoverable since $\chan N^\dagger(M) = 0$ and $M$ preserves the logical code-space.
Using corollary \ref{cor:correctable_with_superselection}, we may take the commutator of $\V^\dagger M \V$ with both sides of eq. \eqref{gkl}.
\begin{equation}
[ \V^\dagger M \V, \V^\dagger E_n^\dagger P_j E_m \V] = 
 \sum_i c_{ijnm} [\V^\dagger M \V, \V^\dagger P_i \V].
\end{equation}
Using that $M$ commutes with $\V\V^\dagger$ and $P_j$, we find that the RHS vanishes for all $n, m$, whereas the LHS remains non-trivial for $n \neq m$.}

\new{This does not provide an explicit computation of the non-unit optimal recovery fidelity.
One possibility for this is to evaluate the RHS expression of eq. \ref{optimalrecovery} taking $\chan M = \chan P$ .
In order to obtain an upper bound on worse case fidelity $F_W$ or entanglement fidelity $F_\rho$, it is sufficient to consider in eq. \ref{eq:worse_fidelity} an initial state $\rho$ presenting maximal correlation of $M$ with an environment (one bit).
We see that $\chan S \widehat{\chan P \chan M}$ will completely break these correlation, which may not be reconstructed by $\chan S$.}

\new{The map $\chan N$ is composed of low weight local Kraus operators suggesting that the Majorana chain is not a reasonable quantum code.
However, physical considerations pertaining to locality (see sec.\ref{sec:fermion_strong_loc}) suppress the possibility of such parity flipping jump operators.
This amounts to the physical assumption that the system does not exchange fermions with its environment. 
Including terms which change the system parity is a physical process referred to as {\em quasi-particle poisoning} and has a ruinous effect on the memory properties of the Majorana wire.}

\new{As shown in corollary \ref{main1}, we may assume that the global recovery begins by performing a projective measurement of the parity.
Furthermore, from corollary \ref{Fermionic_physical_channels} we may decompose the physical channel $\chan N$ into the sum of two terms, a parity preserving channel $\chan N_+$ and a parity flipping channel $\chan N_-$ by grouping the Kraus operators w.r.t. their eigenvalue under $C$ conjugation.
By further assuming that $P_i\V =\V$ for some $i$ (as done in \cite{Pastawski2017b}) the problem of finding a recovery map for $\chan N$ simplifies to the problem of finding independent recovery maps $\chan R_+ = \chan R P_i$ and $\chan R_- = \chan R P_{\bar i}$ for $\chan N_+$ and $\chan N_-$.
Indeed, in section \ref{sec:MajoranaRingStrongLocality} we will focus on correcting for $\chan N_+$ which admits a much more favorable treatment if geometric locality is required on the noise Kraus operators. }

\section{Local reversibility \label{Sec:RevLocal}}

In Refs.~\cite{flammia2016}, the unconstrained dual optimization relation Eq.~\eqref{duality} for $\chan M = \id$ was adapted to characterize local reversibility, i.e., with the constraint that the recovery channel be local to a subsystem. 
This was later applied \cite{Kim2016}, to study local recoverability in the setting of an isometry $W$ defined by a MERA circuit \cite{Vidal2008}.
In this section, we generalize the notion of local recovery to a setting where the physical local algebra can be arbitrary $\dagger$-algebras extending the information-disturbance approach of Refs. \cite{beny2010, beny2011}.

\subsection{Local channels}
\label{local_chan}

Let us consider a {\em local net of algebras}, i.e., a map assigning each region of space $\omega$ to a $\dagger$-subalgebra $\alg A_{\omega}$ such that $\alg A_{\omega} \subseteq \alg A_{\omega'}$ whenever $\omega \subseteq \omega'$, and $\alg A_{\omega}$ and $\alg A_{\omega'}$ commute whenever $\omega \cap \omega' = \emptyset$. Also we assume that $\one \in \alg A_{\omega}$ for all $\omega$.

A channel $\chan N$ local to $\omega$ should be such that $\chanh N(\alg A_\omega) \subseteq \alg A_\omega$, so that an observer with access to $\omega$ cannot learn about observables which are not local to $\omega$ due to the action of the channel. In addition, $\chanh N$ should {\em fix} the observables $\alg A_{\omega^c}$ local to the complement set $\omega^c$, so that an observer having access to $\omega^c$ cannot learn that anything happened.

One may also require that additional observables be fixed, up to the full commutant of $\alg A_\omega$ (which is not in general equal to $\alg A_{\omega^c}$, as can be seen for fermions). We want to be agnostic towards such choices. Hence we define locality generally as follows:

\begin{definition}
\label{local_channel}
We say that channel $\chan N$ (from $\hil H$ to itself) is local to an algebra $\alg A$, with effective complement $\alg B \subseteq \alg A'$ if $\chanh N(\alg A) \subseteq \alg A$ and $\chanh N$ fixes $\alg B$, i.e., $\chanh N(B) = B$ for all $B \in \alg B$.
If $\alg B = \alg A'$ (commutant of $\alg A$), we say that $\chan N$ is {\em strongly local}.
\end{definition}

If $\alg A = \alg A_\omega$, the effective complement $\alg B$ is to be distinguished from the algebra $\alg A_{\omega^c}$ on the complementary region. We may set $\alg B = \alg A_{\omega^c}$ in the above definition, in which case we would say that the channel is {\em weakly local}.

In order to better understand the implications of this definition, we need the following known fact:
\begin{lemma}
\label{fixed-point}
If a channel $\chan N(\rho) = \sum_i E_i \rho E_i^\dagger$ {\em fixes} a  $\dagger$-algebra $\alg B$ (i.e., $\chanh N(B)=B$ for all $B \in \alg B$) then
$[E_i,B] = 0$ for all $B \in \alg B$ and all $i$.
\end{lemma}
\proof
Given that any finite-dimensional $\dagger$-algebra is spanned by its projectors, we only need to show this for projectors $P$ in $\alg B$. 
We have $\sum_i E_i^\dagger P E_i = P$.  Multiplying this equation on both side by $P^\perp = \one - P$ we get $\sum_i P^\perp E_i^\dagger P E_i P^\perp = 0$. Taking the expectation value with respect to any vector, we can deduce that for all $i$, $P E_i P^\perp = 0$. Similarly, since $P^\perp$ is in $\alg B$, we have $P^\perp E_i P = 0$. Together this implies $P E_i = E_i P$.
\qed

This implies that, requiring a channel to be local to $\alg A$ with maximal complement $\alg A'$ is equivalent to asking for its Kraus operators to lie in $\alg A = \alg A''$ (via the double-commutant theorem). 

Furthermore, we also have that for $\chan N$ local to $\alg A$ with complement $\alg B$ then
\begin{equation}
\label{extension_equ}
\chanh N(AB) = \chanh N(A) B
\end{equation}
for all $A \in \alg A$ and $B \in \alg B$.

This also allows us to characterize the local channels in a way that will be useful below. Let $\chan P_{\alg B'}$ be the projector on the commutant $\alg B'$. We can implement it as an integral over that Haar measure in the group of unitary operators within $\alg B$:
\begin{equation}
\chan P_{\alg B'}(X) = \int_{U \in \alg B} U X U^\dagger dU,
\end{equation}
or using a unitary $1$-design $U_i \in \alg B$, $i = 1,\dots,n$:
\begin{equation}
\chan P_{\alg B'}(X) = \frac 1 n \sum_i U_i X U_i^\dagger.
\end{equation}
Hence, we have the Stinespring dilation
\begin{equation}
\chan P_{\alg B'}(X) = V_{\alg B}^\dagger (\one \otimes X) V_{\alg B},
\end{equation}
with 
\begin{equation}
V_{\alg B} = \frac 1 {\sqrt n} \sum_i \ket i \otimes U_i.
\end{equation}
Suppose $\chan N(\rho) = \sum_i E_i \rho E_i^\dagger$, and hence has a dilation isometry $V_{\chan N} = \sum_i E_i \otimes \ket i$: $\chan N^\dagger(X) = V_{\chan N}^\dagger (X \otimes \one) V_{\chan N}$. If $\chan N^\dagger$ fixes $\alg B$, then since $[E_i,U_j] = 0$, we obtain
\begin{equation}
\label{commuting_of_dilations}
(V_{\alg B} \otimes \one) V_{\chan N} = (\one \otimes V_{\chan N}) V_{\alg B}, 
\end{equation}
because, when extended, both sides are proportional to
\begin{equation}
\sum_{ij} \ket i \otimes U_i E_j \otimes \ket j = \sum_{ij} \ket i \otimes E_j U_i  \otimes \ket j.
\end{equation}
Tracing-out the third tensor factor on both sides of the equation yields
\begin{equation}
\label{commute_with_partialexp}
\chan V_{\alg B} \circ \chan N = (\id \otimes \chan N)\circ \chan V_{\alg B},
\end{equation}
where we used $\chan V_{\alg B}(\rho) = V_{\alg B} \rho V_{\alg B}^\dagger $.
Similarly, tracing-out the first tensor factor yields
\begin{equation} 
\label{commute_comp}
(\cchan P_{\alg B'} \otimes \id) \circ \chan V_{\chan N} = (\id \otimes \cchan N)\circ \chan V_{\alg B}.
\end{equation}

\subsection{Local complementary channels}
\label{lcc}

In order to generalize \mytheorem~\ref{thmduality}, we need a notion of local complementary channel.
\begin{definition}
Let $\chan N$ be a channel local to $\alg A$ with complement $\alg B$. Given a Stinespring dilation $\chanh N(X) = V^\dagger (X\otimes \one_E) V$, where $V$ is an isometry: $V^\dagger V = \one$,
We define a corresponding {\em local complementary channel} $\lcchan{\alg B}{N}$ by
\begin{equation}
\lcchanh{\alg B}{N}(B \otimes E) = V^\dagger (\chan P_{\alg B}(B) \otimes E) V.
\end{equation}
\end{definition}

This definition will be justified by the fact that such local complementary channel appears naturally in \mytheorem~\ref{main2a}, where it plays the same role as the normal complementary channel in \mytheorem~\ref{thmduality}.

We can, nevertheless motivate it intuitively as follows. Consider a channel from Alice to Bob. 
The complementary channel represents all the information that can possibly be recovered by a third party---the ``environment''---simultaneously to Bob receiving his information. 
If, however, Bob cannot access the subsystem defined by the algebra $\alg B$ in the output of the channel (because he is a local observer), then that system should also be counted as part of the environment.

A local complementary channel can be expressed in terms of a standard complementary channel through
\begin{equation}
\label{compequ}
\lcchan{\alg B }N = \comp {\chan P_{\alg B'} \chan N},
\end{equation}
where $\chan P_{\alg B'}$ denotes the projector on the commutant $\alg B'$ of the algebra $\alg B$.
This follows from the fact that $\comph {\chan P_{\alg B'} \chan N}(X \otimes E) = V^\dagger (\chan P_{\alg B}(X) \otimes E) V$, where $V$ is an isometry such that $\chanh N(Y) = V^\dagger (Y \otimes \one) V$.

\subsection{Condition for local reversibility}

We now have the tools needed to generalize \mytheorem~\ref{thmduality} to local channels. Let us first consider only the constraint that a local channel must fix some algebra. This is also equivalent to considering only strong locality, since if $\chan N^\dagger$ fixes $\alg B$, then it must also map $\alg B'$ into $\alg B'$ (since its Kraus operators must then all belong to $\alg B'$).

\begin{theorem}
\label{main2a}
Let $\chan N$ and $\chan M$ be two channels such that both $\chan N^\dagger$ and $\chan M^\dagger$ fix the algebra $\alg B$. Then
\begin{equation}
\label{fixedduality}
\begin{split}
\max_{\text{$\chan R^\dagger$ fixes $\alg B$}} F(\chan R \chan N, \chan M) 
&= \max_{\text{$\chan S^\dagger$ fixes $\alg B$}} F(\lcchan{\alg B}{N}, \chan S \lcchan{\alg B}{M}).\\
\end{split}
\end{equation}
\end{theorem} 
\proof  
Let $\chan V_{\alg B}$ be an isometry dilating the projector on $\alg B'$ as in Section~\ref{local_chan}. Using the fact that the fidelity is invariant under post-processing by an isometry, and then Eq.~\eqref{commute_with_partialexp} ,
\begin{equation}
\begin{split}
F(\chan R \chan N, \chan M) &= F(\chan V_{\alg B} \chan R \chan N, \chan V_{\alg B}\chan M) \\
&= F((\id \otimes \chan R \chan N) \chan V_{\alg B} , (\id \otimes \chan M) \chan V_{\alg B}). \\
\end{split}
\end{equation}
If we use the entanglement fidelity $F_\rho$, then this last term is just
\begin{equation}
F_{\chan P_{\alg B'}(\rho)}(\chan R \chan N, \chan M). \\
\end{equation}
We can then apply \mytheorem~\ref{thmduality} to this quantity, 
to obtain
\begin{equation}
\max_{\chan R} F_{\chan P_{\alg B'}(\rho)}(\chan R \chan N, \chan M) 
= \max_{\chan S} F_{\chan P_{\alg B'}(\rho)}(\cchan N, \chan S \cchan M).
\end{equation}
This also works for the worst-case fidelity $F_\V$, but we need a cosmetically stronger version of \mytheorem~\ref{thmduality} that would hold when minimzing over states of the form $\chan P_{\alg B'}(\rho)$ rather than all states. But since this set is also convex, the proof in Ref.~\cite{beny2010} (or Ref.~\cite{beny2011} with more details) works unchanged. 
 
Using Eq.~\eqref{commute_comp}, we obtain, both for $F$ replaced by $F_\rho$ or $F_\V$,
\begin{equation}
\begin{split}
&F((\id \otimes \chan R \chan N) \chan V_{\alg B} , (\id \otimes \chan M) \chan V_{\alg B})\\
&\quad \quad = F((\id \otimes \cchan N) \chan V_{\alg B}, (\id \otimes \chan S \cchan M) \chan V_{\alg B})\\
&\quad \quad = F((\cchan P_{\alg B'} \otimes \id)  \chan V_{\chan N}, (\cchan P_{\alg B'} \otimes \chan S) \chan V_{\chan M})\\
&\quad \quad = F((\chan P_{\alg B} \otimes \id)  \chan V_{\chan N}, (\chan P_{\alg B} \otimes \chan S) \chan V_{\chan M})\\
&\quad \quad = F(\lcchan{\alg B}{N}, (\id \otimes \chan S) \lcchan{\alg B}{M}).\\
\end{split}
\end{equation}
In the second-to-last step, we used the fact that $\chan P_{\alg B}$ is complementary to $\chan P_{\alg B'}$, and hence is equivalent to $\cchan P_{\alg B'}$ up to a reversible post-processing, which cannot change the value of the fidelity. 
\qed

We can now combine this with the approach of Section~\ref{rev_const}, to obtain a dual condition for the local correctability of a channel. Where locality is defined with respect to an algebra $\alg A$ with commuting complement $\alg B$.
\begin{corollary}
\label{main2b}
Let $\chan P$ be the projector on $\alg A \vee \alg B$ (the algebra generated by $\alg A \cup \alg B$) and $\chan N$ a channel local to $\alg A$ with complement $\alg B$, then  
\begin{equation}
\label{localduality}
\begin{split}
\max_{\text{$\chan R$ local}} F(\chan R \chan N, \chan P) 
&= \max_{\text{$\chan S^\dagger$ fixes $\alg B$}} F(\lcomp{\alg B}{ \chan P \chan N}, \chan S \lcchan{\alg B} P),\\
\end{split}
\end{equation}
where the maximization on the left-hand side is over channels $\chan R$ which are local in the same sense as for $\chan N$.
\end{corollary} 
\proof  
From \mytheorem~\ref{main2a}, the right hand side of Eq.~\eqref{localduality} is equal to
\begin{equation}
\max_{\text{$\chan R^\dagger$ fixes $\alg B$}} F(\chan R \chan P \chan N, \chan P).
\end{equation}
But the set of channels $\chan R \chan P$ where $\chan R$ reaches the maximum in this expression must include some which are local, since by monotonicity of the fidelity, the local map $\chan P \chan R \chan P$ can only perform better than $\chan R \chan P$. Also, if $\chan R$ is a local channel which maximizes $F(\chan R \chan N, \chan P)$, then so does $\chan P \chan R \chan P$ since then 
\begin{equation}
F(\chan P \chan R \chan P \chan N, \chan P) = F(\chan P \chan R \chan N, \chan P) \ge F(\chan R \chan N, \chan P).
\end{equation}
\qed

We observe that 
\(
\lcomp{\alg B}{ \chan P \chan N} = \comp{ \chan P_{\alg B^c} \chan N}
\)
and we can use
\(
\lcchan{\alg B} P = \chan P_{(\alg B^c)'},
\)
where $\alg B^c = \alg B' \cap (\alg A \vee \alg B)$ is the commutant of $\alg B$ {\em relative to} $\alg A \vee \alg B$.

We can also be more explicit about the structure of these algebras.
Since $\alg A$ and $\alg B$ commute with each other, the intersection 
\begin{equation}
\alg I := \alg A \cap \alg B
\end{equation}
is a commutative $\dagger$-algebra. Hence it is spanned by a complete family of projectors $P_i$, $i=1,\dots,n$, such that $P_i P_j = \delta_{ij} P_i$ and $\sum_i P_i = \one$.

The algebras $\alg A$ and $\alg B$ must be block-diagonal in terms of the sectors $\hil H_i = P_i \hil H$, $i=1,\dots,N$, and, since they commute, they must take the form
\begin{equation}
\alg A = \bigoplus_{i=1}^N \alg A_i \otimes \one_{m_i} \;\;\text{and}\;\; \alg B = \bigoplus_{i=1}^N \one_{n_i} \otimes \alg B_i
\end{equation}
where the $i$th term of the direct sum is supported on the sector $\hil H_i$, and $\alg A_i$ and  $\alg B_i$ are themselves $\dagger$-algebras.

Then the relative commutant of $\alg B$ in $\alg A \vee \alg B$ has the form 
\begin{equation}
\alg B^c = \bigoplus_i \alg A_i \otimes \mc Z(\alg B_i),
\end{equation}
where $\mc Z(\alg B_i) = \alg B_i \cap \alg B_i'$ is a commutative algebra: the center of $\alg B_i$.

\subsection{Example: standard tensor product}
\label{locex1}

To see how this works, let us first consider the meaning of Eq.~\eqref{localduality} being equal to unity (maximal) for channels which are local in the usual sense of a tensor product of Hilbert space, say system $\hil H_A \otimes \hil H_B$. The algebras are $\alg A = \ops(\hil H_A) \otimes \one$ and $\alg B = \one \otimes \ops(\hil H_B)$. 
The channel $\chan N$ being local to system $A$ implies $\chan N = \chan N_A \otimes \id$. 

The problem considered here is whether the specific noise channel $\chan N$ can be corrected by a channel acting within $A$. This is not to be confused with the more standard  problem of recovering arbitrary noise channels on $A$ without restricting the locality of the recovery operation.


Let us define $\chan D_{\sigma}$ to be a fully depolarizing channel, with constant output state $\sigma$: $\chan D_\sigma(\rho) = \sigma\, \tr \rho$ for all $\rho$.

We observe that $\alg B^c = \alg A = \chan B(\hil H_A) \otimes \one$ and $(\alg B^c)' = \alg B = \one \otimes \chan B(\hil H_B)$. Hence 
\begin{equation}
\chan P_{\alg B^c} = \id \otimes \chan D_{\one / d_B} \quad \text{and} \quad \lcchan {\alg B} P = \chan P_{(\alg B^c)'} = \chan D_{\one / d_A} \otimes \id, 
\end{equation}
where $d_A$ and $d_B$ are the dimensions of $\hil H_A$ and $\hil H_B$ respectively.
It follows that 
\begin{equation}
\lcomp{\alg B}{ \chan P \chan N} = \comp{\chan P_{\alg B^c} \chan N} = \comp{ \chan N_A \otimes \chan D_{\one / d_B}} = \cchan N_A \otimes \id.
\end{equation}

Moreover, for any channel $\chan S$, $\chan S \circ \chan D_{\one / d_A} = \chan D_\sigma$, where $\sigma = \chan S(\one / d_A)$.
Hence, \mycorollary~\ref{main2b} tells us that 
\begin{equation}
\max_{\text{$\chan R$}} F(\chan R \chan N_A \otimes \id_B, \id_{AB}) 
= \max_{\sigma} F(\cchan N_A \otimes \id_B,  \chan D_\sigma \otimes \id_B).
\end{equation}

This is almost exactly like the QEC conditions without the locality constraint: the channel is reversible if and only if the local environment gets no information. 
The difference is that the code is defined as a subspace of the joint system $AB$ instead of just system $A$.

To understand this in more detail, let us expand the conditions resulting for the fidelity $F = F_\V$ being maximal on both sides. 
 If we write this condition in terms of the Kraus operators $E_i$ of $\chan N_A$, we obtain that the condition is that there exists a state $\sigma$ such that for all operators $E \otimes B$,
\begin{equation}
\sum_{ij} \V^\dagger (E_i^\dagger E_j \otimes B) \V \bra i E \ket j  = \V^\dagger (\one \otimes B) \V \, \tr(\sigma E).
\end{equation}
Equivalently, for all $i$, $j$, and all $B$, there must be numbers $\lambda_{ij}$ such that 
\begin{equation}
\label{exqubitlocalcond}
\V^\dagger (E_i^\dagger E_j \otimes B) \V = \lambda_{ij} \V^\dagger (\one \otimes B) \V.
\end{equation}
This can also be formulated in terms of matrix elements $B = \ketbra mn$: for all $i,j,m,n$, there must exist $\lambda_{ij} \in \mathbb C$ such that
\begin{equation}
\V_m^\dagger E_i^\dagger E_j \V_n = \lambda_{ij} \V_m^\dagger \V_n,
\end{equation}
where $\V_n := (\one \otimes \bra n) \V$.

\subsection{Example: fermions and strong locality}\label{sec:fermion_strong_loc}

Let us unpack \mycorollary~\ref{main2b} for fermions, and with the strong locality requirement $\alg B = \alg A'$. 

Firstly, if $\alg B = \alg A'$, then $\alg B^c = \alg A$, and $\lcomp{\alg B}{ \chan P \chan N} = \comp{\chan P_{\alg A} \chan N}$. Similarly, $\lcchan {\alg B} P = \chan P_{\alg A'}$. 

Let $\alg A = \alg A_\omega$ be a local algebra for fermions for a region $\omega$. 
A charge generating the center $\alg I = \alg A \cap \alg A'$ is the parity $C_\omega$ of the number of fermions in the region $\omega$. 
\new{Note that the same statement holds if $\omega$ and $\omega^c$ are subsets indexing  Majorana modes.} 
The commutant is just $\alg A' = \alg I \vee \alg A_{\omega^c} = {\rm span}\{C, \alg A_{\omega^c}\}$, where $C$ is the global parity observable.

\mycorollary~\ref{main2b} tells us that the optimal fidelity for local reversal of $\chan N$ is equal to 
\(
\max_{\chan S} F(\comp{\chan P_{\alg A} \chan N}, \chan S \chan P_{\alg A'})
\)
where $\chan S^\dagger$ must fix $\alg A'$. 

But the channel $\chan S$ needs only be defined on $\alg A'$ since it acts after a projection on it. Since its adjoint must fix $\alg A'$, it can be assumed to be of the form
\begin{equation}
\chan S^\dagger(B \otimes E) = \sum_{k\in \{+,-\}} \tr(\rho_k E) P_k B P_k,
\end{equation}
for any $B \in \alg A'$ and $E$ acting on the environment from the dilation of $\chan N$. The only freedom are the two arbitrary fixed states $\rho_k$, $k \in \{+, -\}$. Here, $P_{+}$ and $P_{-}$ are the projectors on even and odd parity for the region $\omega$. .

If $E_i$ are Kraus operators for $\chan N$, the code $\V$ is then exactly locally correctable if and only if there exist two states $\rho_{\pm}$ such that for all $B \in \alg A'$, and all $E$,
\begin{equation}
\begin{split}
\sum_{ij} \bra i  E \ket j\,\V^\dagger E_i^\dagger B  E_j \V &= \V^\dagger \chan S^\dagger(B \otimes E) \V\\
&= \sum_{k\in \{+,-\}} \tr(\rho_k E) \V^\dagger P_k B P_k \V,\\
\end{split}
\end{equation}
or, equivalently, for all $i$, $j$, any parity $k$, and for all $B \in \alg A'$, there exist $\lambda_{ijk} \in \mathbb C$ such that
\begin{equation}\label{eq:FermionRecovery}
\V^\dagger E_i^\dagger E_j B_k \V  = \lambda_{ijk} \V^\dagger B_k \V,
\end{equation}
where $B_k = P_k B P_k$ commute with the operators $E_i$.

The only difference with Eq.~\eqref{exqubitlocalcond} is that $B$ is restricted to the algebra $\alg A'$ which is a direct sum of two factors, corresponding to the two values for the parity of the region $\omega$.

For instance, if the code defined by $\V$ is restricted to states of a fixed parity, then the conditions are of the same form as in the Section~\ref{locex1}.

We see that these conditions do not reduce to those of Section~\eqref{ex_rev_sup} when $\omega$ is the whole system. This is because the strong locality condition is non-trivial in this limit, as the channel is still required to fix the global parity operator, which commutes with $\alg A_\Omega$. This is what makes the above conditions simpler.

\subsection{Example: Majorana ring}\label{sec:MajoranaRingStrongLocality}

In this section we extend the treatment of the Majorana ring, introduced in Section~\ref{sec:MajoranaRing}. In addition to physicality, we require both the noise and the recovery map ($\chan N$  and $\chan R$) to be strongly local to the algebra $\alg A_L$, with $L$ corresponding to some subset of Majorana modes on the ring.
As defined in \ref{local_channel}, this means that $\chan N^\dagger(B) = B$ for all operators $B \in \alg A'_L$.

According to Lemma~\ref{fixed-point}, and the double commutant theorem, the noise map $\chan N$ will have all Kraus operators $E_j \in \alg A_L$ (i.e. even weight polynomials of Majorana operators with indices in $L$). 
This amounts to the physical assumption that the system $L$ does not exchange fermions with the rest of the system or the environment.  

For simplicity, let us consider the case where $L$ includes the whole chain. In that case, the strong locality condition simply precludes quasi-particle poisoning. 
Since, $B_k \propto P_k$ in Eq.~\eqref{eq:FermionRecovery}, if we also assume that the code projector $WW^\dagger$ commutes with parity, then one can simply remove the terms $B_k$ from the equation. The necessary and sufficient conditions we obtain are (formally) equal to the standard Knill-Laflamme conditions.

In this context, a sufficient condition for the noise to satisfy Eq.~\eqref{eq:FermionRecovery} is that the noise operators $E_j$ do not involve ``distant'' Majorana modes (geometric locality).
Specifically, we assume that the Majorana modes appearing in the monomial expansion of $E_j$ are supported on an index interval of length at most $d/2$, where $d$ is the number of mode in the shortest region $I_j$.

(To be clear, this geometric locality of the noise operators does not imply locality of the noise channel as we have defined it in this document, because different noise operators of the channel may be local to different regions).

In order to show that the KL conditions are 
 satisfied for these error operators, we may expand $E_j$ in terms of even Majorana monomials (these span the algebra). 
The geometric locality of these operators then guarantees that they act non-trivially on at most one of the unpaired Majorana modes. 
But because the degree of the monomial is even, there must be an unmatched Majorana operator acting on a paired mode, hence mapping it to an excited state. 
Since $E_i^\dagger$ and $E_j$ in the expression $E_i^\dagger E_j$ cannot overlap if they each act on distinct unpaired Majorana modes in $\omega$, $E_i^\dagger E_j$ yield mutually orthogonal states. 
If both act on the same mode in $\omega$, canceling on it, then they either commute with the projector and are proportional to identity or are fully off-diagonal with respect to it.
Either of these cases satisfies Eq. ~\eqref{eq:FermionRecovery}.

\section{Conclusions}
We have extended the result of Refs. \cite{beny2010, beny2011} on approximate channel recovery to settings where the channels are restricted by symmetry and/or locality constraints.
The results obtained preserve the elegant duality structure of the information disturbance trade-off.
Namely, they take the form of a dual optimization representing the information which is not accessible to the environment. 
Although we do not have a general solution for the dual optimization, it is often simpler, and has a solution in important specific situations. 
As the duality itself has already proven useful conceptually~\cite{kretschmann2008,Almheiri2014,pastawski2016,flammia2016,Harlow2017},
we expect that our results will be widely applicable to the QEC aspects of symmetry protected topological phases as well as in some more realistic realizations of holography.

\section{Acknowledgements}
CB is thankful to Reinhard Werner and Tobias Osborne for discussions about channel locality. FP and ZZ would like to acknowledge Henrik Wilming and Szil\'ard Szalay for useful and interesting discussions. CB was supported by the National Research Foundation of Korea (NRF-2018R1D1A1A02048436). FP acknowledges the Alexander von Humboldt foundation. ZZ was supported by Hungarian National Research, Development and Innovation Office NKFIH (Contracts No. K124351,
K124152, and K124176) and the J\'anos Bolyai Scholarship of the Hungarian Academy of Sciences.

\bibliographystyle{unsrt}

\bibliography{complete.bib} 

\begin{thebibliography}{10}

\bibitem{Wen2004}
Xiao-Gang Wen.
\newblock {\em {Quantum field theory of many-body systems : from the origin of
  sound to an origin of light and electrons}}.
\newblock Oxford University Press, 2004.

\bibitem{Zeng2015}
Bei Zeng, Xie Chen, Duan-Lu Zhou, and Xiao-Gang Wen.
\newblock {Quantum Information Meets Quantum Matter -- From Quantum
  Entanglement to Topological Phase in Many-Body Systems}.
\newblock aug 2015.
\newblock \href{http://arxiv.org/abs/1508.02595}{(arXiv:1508.02595)}.

\bibitem{schuch2004}
Norbert Schuch, Frank Verstraete, and J~Ignacio Cirac.
\newblock Nonlocal resources in the presence of superselection rules.
\newblock {\em Physical Review Letters}, 92(8):087904, 2004.
\newblock
  \href{https://arxiv.org/abs/quant-ph/0310124}{(arXiv:quant-ph/0310124)}.

\bibitem{Banuls2007}
Mari-Carmen Ba{\~n}uls, J~Ignacio Cirac, and Michael~M Wolf.
\newblock Entanglement in fermionic systems.
\newblock {\em Physical Review A}, 76(2):022311, 2007.
\newblock \href{https://arxiv.org/abs/0705.1103}{(arXiv:0705.1103)}.

\bibitem{donnelly2012}
William Donnelly.
\newblock Decomposition of entanglement entropy in lattice gauge theory.
\newblock {\em Physical Review D}, 85(8):085004, 2012.
\newblock \href{https://arxiv.org/abs/1109.0036}{(arXiv:1109.0036)}.

\bibitem{casini2014}
Horacio Casini, Marina Huerta, and Jos{\'e}~Alejandro Rosabal.
\newblock Remarks on entanglement entropy for gauge fields.
\newblock {\em Physical Review D}, 89(8):085012, 2014.
\newblock \href{https://arxiv.org/abs/1312.1183}{(arXiv:1312.1183)}.

\bibitem{vanacoleyen2016}
Karel Van~Acoleyen, Nick Bultinck, Jutho Haegeman, Michael Marien, Volkher~B
  Scholz, and Frank Verstraete.
\newblock Entanglement of distillation for lattice gauge theories.
\newblock {\em Physical Review Letters}, 117(13):131602, 2016.
\newblock \href{https://arxiv.org/abs/1511.04369}{(arXiv:1511.04369)}.

\bibitem{Witten2018}
Edward Witten.
\newblock {Notes on Some Entanglement Properties of Quantum Field Theory}.
\newblock March 2018.
\newblock \href{https://arxiv.org/abs/1803.04993}{(arXiv:1803.04993)}.

\bibitem{Chen2011}
Xie Chen, Zheng-Cheng Gu, and Xiao-Gang Wen.
\newblock {Classification of gapped symmetric phases in one-dimensional spin
  systems}.
\newblock {\em Physical Review B}, 83(3):035107, jan 2011.
\newblock \href{ https://arxiv.org/abs/1008.3745}{(arXiv:1008.3745)}.

\bibitem{Chen2013}
Xie Chen, Zheng-Cheng Gu, Zheng-Xin Liu, and Xiao-Gang Wen.
\newblock {Symmetry protected topological orders and the group cohomology of
  their symmetry group}.
\newblock {\em Physical Review B}, 87(15):155114, apr 2013.
\newblock \href{ https://arxiv.org/abs/1106.4772}{(arXiv:1106.4772)}.

\bibitem{Senthil2015}
T.~Senthil.
\newblock {Symmetry-Protected Topological Phases of Quantum Matter}.
\newblock {\em Annual Review of Condensed Matter Physics}, 6(1):299--324, mar
  2015.
\newblock \href{ https://arxiv.org/abs/1405.4015}{(arXiv:1405.4015)}.

\bibitem{Bravyi2010}
Sergey Bravyi, Matthew~B. Hastings, and Spyridon Michalakis.
\newblock {Topological quantum order: Stability under local perturbations}.
\newblock {\em Journal of Mathematical Physics}, 51(9):093512, sep 2010.
\newblock \href{ https://arxiv.org/abs/1001.0344}{(arXiv:1001.0344)}.

\bibitem{Michalakis2013}
Spyridon Michalakis and Justyna~P. Zwolak.
\newblock {Stability of Frustration-Free Hamiltonians}.
\newblock {\em Communications in Mathematical Physics}, 322(2):277--302, jul
  2013.
\newblock \href{https://arxiv.org/abs/1109.1588}{(arXiv:1109.1588)}.

\bibitem{Kitaev2001}
A~Yu Kitaev.
\newblock Unpaired majorana fermions in quantum wires.
\newblock {\em Physics-Uspekhi}, 44(10S):131--136, oct 2001.
\newblock \href{
  https://arxiv.org/abs/cond-mat/0010440}{(arXiv:cond-mat/0010440)}.

\bibitem{Bravyi:2010de}
Sergey Bravyi, Barbara~M. Terhal, and Bernhard Leemhuis.
\newblock {Majorana Fermion Codes}.
\newblock {\em New J. Phys.}, 12:083039, 2010.

\bibitem{Roberts2018}
Sam Roberts and Stephen~D. Bartlett.
\newblock {Symmetry-protected self-correcting quantum memories}.
\newblock may 2018.
\newblock \href{ http://arxiv.org/abs/1805.01474}{(arXiv:1805.01474)}.

\bibitem{Mourik2012}
V~Mourik, K~Zuo, S~M Frolov, S~R Plissard, E~P A~M Bakkers, and L~P
  Kouwenhoven.
\newblock {Signatures of Majorana fermions in hybrid
  superconductor-semiconductor nanowire devices.}
\newblock {\em Science (New York, N.Y.)}, 336(6084):1003--7, may 2012.
\newblock \href{ https://arxiv.org/abs/1204.2792}{(arXiv:1204.2792)}.

\bibitem{Almheiri2014}
Ahmed Almheiri, Xi~Dong, and Daniel Harlow.
\newblock Bulk locality and quantum error correction in ads/cft.
\newblock {\em Journal of High Energy Physics}, 2015(4):1--34, 2014.
\newblock \href{https://arxiv.org/abs/1411.7041}{(arXiv:1411.7041)}.

\bibitem{pastawski2015}
Fernando Pastawski, Beni Yoshida, Daniel Harlow, and John Preskill.
\newblock Holographic quantum error-correcting codes: Toy models for the
  bulk/boundary correspondence.
\newblock {\em Journal of High Energy Physics}, 2015(6):149, 2015.
\newblock \href{https://arxiv.org/abs/1503.06237}{arXiv:1503.06237}.

\bibitem{pastawski2016}
Fernando Pastawski and John Preskill.
\newblock Code properties from holographic geometries.
\newblock {\em Physical Review X}, 7(2):021022, 2017.
\newblock \href{https://arxiv.org/abs/1612.00017}{(arXiv:1612.00017)}.

\bibitem{Hayden2016}
Patrick Hayden, Sepehr Nezami, Xiao-Liang Qi, Nathaniel Thomas, Michael Walter,
  and Zhao Yang.
\newblock {Holographic duality from random tensor networks}.
\newblock {\em Journal of High Energy Physics}, 2016(11):9, nov 2016.
\newblock \href{https://arxiv.org/abs/1601.01694}{(arXiv:1601.01694)}.

\bibitem{Cotler2017}
Jordan Cotler, Patrick Hayden, Grant Salton, Brian Swingle, and Michael Walter.
\newblock {Entanglement Wedge Reconstruction via Universal Recovery Channels}.
\newblock apr 2017.
\newblock \href{http://arxiv.org/abs/1704.05839}{(arXiv:1704.05839)}.

\bibitem{Mintun2015}
Eric Mintun, Joseph Polchinski, and Vladimir Rosenhaus.
\newblock {Bulk-Boundary Duality, Gauge Invariance, and Quantum Error
  Corrections}.
\newblock {\em Physical Review Letters}, 115(15), jan 2015.
\newblock \href{http://arxiv.org/abs/1501.06577}{(arXiv:1501.06577)}.

\bibitem{Pastawski2017b}
Fernando Pastawski, Jens Eisert, and Henrik Wilming.
\newblock {Towards Holography via Quantum Source-Channel Codes}.
\newblock {\em Physical Review Letters}, 119(2):020501, nov 2017.
\newblock \href{https://arxiv.org/abs/1611.07528}{(arXiv:1611.07528)}.

\bibitem{Beny2007}
C{\'e}dric B{\'e}ny, Achim Kempf, and David~W Kribs.
\newblock {Generalization of quantum error correction via the Heisenberg
  picture}.
\newblock {\em Physical Review Letters}, 98(10):100502, 2007.
\newblock
  \href{https://arxiv.org/abs/quant-ph/0608071}{(arXiv:quant-ph/0608071)}.

\bibitem{knill2000}
Emanuel Knill, Raymond Laflamme, and Lorenza Viola.
\newblock Theory of quantum error correction for general noise.
\newblock {\em Physical Review Letters}, 84(11):2525, 2000.
\newblock
  \href{https://arxiv.org/abs/quant-ph/9908066}{(arXiv:quant-ph/9908066)}.

\bibitem{Kribs2005}
David Kribs, Raymond Laflamme, and David Poulin.
\newblock Unified and generalized approach to quantum error correction.
\newblock {\em Physical review letters}, 94(18):180501, 2005.
\newblock
  \href{https://arxiv.org/abs/quant-ph/0412076}{(arXiv:quant-ph/0412076)}.

\bibitem{beny2010}
C{\'e}dric B{\'e}ny and Ognyan Oreshkov.
\newblock General conditions for approximate quantum error correction and
  near-optimal recovery channels.
\newblock {\em Physical Review Letters}, 104(12):120501, 2010.
\newblock \href{https://arxiv.org/abs/0907.5391}{(arXiv:0907.5391)}.

\bibitem{kretschmann2008}
Dennis Kretschmann, Dirk Schlingemann, and Reinhard~F Werner.
\newblock {The information-disturbance tradeoff and the continuity of
  Stinespring's representation}.
\newblock {\em IEEE transactions on information theory}, 54(4):1708--1717,
  2008.
\newblock
  \href{https://arxiv.org/abs/quant-ph/0605009}{(arXiv:quant-ph/0605009)}.

\bibitem{beny2011}
C{\'e}dric B{\'e}ny and Ognyan Oreshkov.
\newblock Approximate simulation of quantum channels.
\newblock {\em Physical Review A}, 84(2):022333, 2011.
\newblock \href{https://arxiv.org/abs/1103.0649}{(arXiv:1103.0649)}.

\bibitem{zimboras2014}
Zolt{\'a}n Zimbor{\'a}s, Robert Zeier, Michael Keyl, and Thomas
  Schulte-Herbr{\"u}ggen.
\newblock A dynamic systems approach to fermions and their relation to spins.
\newblock {\em EPJ Quantum Technology}, 1(1):11, 2014.
\newblock \href{https://arxiv.org/abs/1211.2226}{(arXiv:1211.2226)}.

\bibitem{petz1988}
D{\'e}nes Petz.
\newblock Sufficiency of channels over von {N}eumann algebras.
\newblock {\em The Quarterly Journal of Mathematics}, 39(1):97--108, 1988.

\bibitem{Knill1997}
Emanuel Knill and Raymond Laflamme.
\newblock {Theory of quantum error-correcting codes}.
\newblock {\em Physical Review A}, 55(2):900--911, feb 1997.
\newblock
  \href{https://arxiv.org/abs/quant-ph/9604034}{(arXiv:quant-ph/9604034)}.

\bibitem{Junge2016}
Marius Junge, Renato Renner, David Sutter, Mark~M Wilde, and Andreas Winter.
\newblock Universal recoverability in quantum information.
\newblock In {\em 2016 IEEE International Symposium on Information Theory
  (ISIT)}, pages 2494--2498. IEEE, 2016.
\newblock \href{https://arxiv.org/abs/1509.07127}{(arXiv:1509.07127)}.

\bibitem{buscemi2016}
Francesco Buscemi, Siddhartha Das, and Mark~M Wilde.
\newblock Approximate reversibility in the context of entropy gain, information
  gain, and complete positivity.
\newblock {\em Physical Review A}, 93(6):062314, 2016.
\newblock \href{https://arxiv.org/abs/1601.01207}{(arXiv:1601.01207)}.

\bibitem{schumacher1996}
Benjamin Schumacher.
\newblock Sending entanglement through noisy quantum channels.
\newblock {\em Physical Review A}, 54(4):2614, 1996.
\newblock
  \href{https://arxiv.org/abs/quant-ph/9604023}{(arXiv:quant-ph/9604023)}.

\bibitem{Beny2007b}
C{\'e}dric B{\'e}ny, Achim Kempf, and David~W Kribs.
\newblock Quantum error correction of observables.
\newblock {\em Physical Review A}, 76(4):042303, 2007.
\newblock \href{https://arxiv.org/abs/0705.1574}{(arXiv:0705.1574)}.

\bibitem{Beny2009}
C{\'e}dric B{\'e}ny.
\newblock Conditions for the approximate correction of algebras.
\newblock In {\em Workshop on Quantum Computation, Communication, and
  Cryptography}, pages 66--75. Springer, 2009.
\newblock \href{https://arxiv.org/abs/0907.4207}{(arXiv:0907.4207)}.

\bibitem{Alicea2010}
Jason Alicea.
\newblock {Majorana fermions in a tunable semiconductor device}.
\newblock {\em Physical Review B}, 81(12):125318, mar 2010.

\bibitem{PhysRevB.85.174533}
Diego Rainis and Daniel Loss.
\newblock Majorana qubit decoherence by quasiparticle poisoning.
\newblock {\em Phys. Rev. B}, 85:174533, May 2012.

\bibitem{Temme2014}
Kristan Temme, Fernando Pastawski, and Michael~J. Kastoryano.
\newblock {Hypercontractivity of quasi-free quantum semigroups}.
\newblock {\em Journal of Physics A: Mathematical and Theoretical},
  47(40):405303, mar 2014.
\newblock \href{ http://arxiv.org/abs/1403.5224}{(arXiv:1403.5224)}.

\bibitem{flammia2016}
Steven~T Flammia, Jeongwan Haah, Michael~J Kastoryano, and Isaac~H Kim.
\newblock Limits on the storage of quantum information in a volume of space.
\newblock {\em Quantum}, 1:4, 2017.
\newblock \href{https://arxiv.org/abs/1610.06169}{(arXiv:1610.06169)}.

\bibitem{Kim2016}
Isaac~H. Kim and Michael~J. Kastoryano.
\newblock {Entanglement renormalization, quantum error correction, and bulk
  causality}.
\newblock {\em Journal of High Energy Physics}, 2017(4), dec 2017.
\newblock \href{http://arxiv.org/abs/1701.00050}{arXiv:1701.00050}.

\bibitem{Vidal2008}
G.~Vidal.
\newblock Class of quantum many-body states that can be efficiently simulated.
\newblock {\em Phys. Rev. Lett.}, 101:110501, 2008.
\newblock
  \href{https://arxiv.org/abs/quant-ph/0610099}{(arXiv:quant-ph/0610099)}.

\bibitem{Harlow2017}
Daniel Harlow.
\newblock {The Ryu-Takayanagi Formula from Quantum Error Correction}.
\newblock {\em Communications in Mathematical Physics}, 354(3):865--912, sep
  2017.
\newblock \href{http://arxiv.org/abs/1607.03901}{arXiv:1607.03901}.

\end{thebibliography}

\end{document}